\documentclass[12pt,reqno, letterpaper]{amsart}
\usepackage{amssymb,epsf}
\usepackage{amstext,amsmath,amsfonts,amscd,amsthm,graphicx}
\usepackage{epsfig}
\usepackage{eufrak}
\usepackage{latexsym}
\usepackage{eucal}
\theoremstyle{amsart}
\newfont{\fnt}{cmsy10}
\newfont{\sss}{cmr10}
\newfont{\azb}{wncyr10}
\newfont{\azbit}{wncyi10}
\theoremstyle{definition}

\theoremstyle{plain}

\newtheorem{vt}{Theorem}
\newtheorem{lm}{Lemma}

\newtheorem{prp}{Proposition}
\theoremstyle{definition}
\newtheorem{pz}{Remark}

\begin{document}
\title[On symmetries and conservation laws of the Majda--Biello system]{
{\protect\vspace*{-1.6cm}}
On symmetries and conservation laws of the Majda--Biello system}
\author{Ji\v{r}ina Vodov\'a-Jahnov\'{a}}
\address{Mathematical Institute, Silesian University in Opava, Na Rybn\'{i}\v{c}ku~1, 746~01 Opava, Czech Republic}
\email{Jirina.Vodova@math.slu.cz, Jirina.Jahnova@math.slu.cz}
\begin{abstract} 
\looseness=-1
In 2003, A.J. Majda and J.A. Biello derived and studied the so-called reduced equations for equatorial baroclinic-barotropic waves, to which we refer as to the Majda--Biello system. The equations in question describe the nonlinear interaction of long-wavelength equatorial Rossby waves and barotropic Rossby waves with a significant midlatitude projection in the presence of suitable horizontally and vertically sheared zonal mean flows.

Below we present a Hamiltonian structure for Majda--Biello system and describe all generalized symmetries and conservation laws for the latter. It turns out that there are only three symmetries corresponding to $x$-translations, $t$- translations and to a scaling of $t$, $x$, $u$ and $v$, and four conservation laws, one of which is associated to the conservation of energy, the second conserved quantity is just the Hamiltonian functional and the other two are Casimir functionals of the Hamiltonian operator admitted by our system. Our result provides {\em inter alia} a rigorous proof of the fact that the Majda--Biello system has just the conservation laws mentioned in the paper by Majda and Biello.
\end{abstract}
\maketitle
\section{Introduction}

The problem of short-time and long-time weather forecasts and the prediction of climate variations on seasonal, yearly and decade time scales is currently of great interest for both climatologists and meteorologists, as well as mathematicians and physicists. One of the recent discoveries in this field is the profound impact of variations in the tropics on long-range seasonal forecasting and climate on the entire globe. The problem of dynamics of the equatorial atmosphere is intensively studied e.g.\ in \cite{majda3, kiladis, majda2, majda4, majda1, majda5, majda6, wang}. In \cite{majda1}, the so-called two-layer-equatorial $\beta$-plane equations (TLEPE) for the barotropic (propagating in the poleward direction) and baroclinic (i.e., trapped in the vicinity of and propagated along the equator) horizontal velocity and pressure have been derived. When the barotropic-baroclinic cross-interaction terms in the TLEPE equations are neglected, then these equations split naturally into nonlinear rotating 2D incompressible Euler equations for the barotropic flow and a rotating linear shallow-water system for the baroclinic mode (see e.g.\ \cite{majda3, majda2}; cf.\ also \cite{burde} and references therein for nonlinear system describing shallow-water waves) respectively, the former possesing the wave-like solutions known as barotropic Rossby waves that travel northward and southward towards the midlatitudes and therefore are an interesting (and it is believed that main) engine for energy exchange between the tropics and extra-tropics, the latter possessing wave-like solutions that are trapped near the equator.\looseness=-1

In \cite{majda1}, the TLEPE equations (containing the cross-interactions terms) are studied, the so-called long-wave-scaled-equatorial baroclinic-barotropic equations (LWSEBB) are derived from them, and small-amplitude weakly nonlinear solutions  of the LWSEBB equations in the form of an asymptotic expansion with parameters assuring the resonance of dispersionless packets are constructed. The solvability conditions yield the
the reduced equations for equatorial baroclinic-barotropic waves, to which we hereinafter refer to as to the
Majda--Biello system:
\begin{eqnarray}\label{syst}
u_t&=& du_{xxx}-vu_x-uv_x\nonumber\\
v_t&=&v_{xxx}-uu_x.
\end{eqnarray}
Here $d=1-\frac{1}{(2m+1)^2}$ is a parameter depending on the meridional index $m\in\mathbb{Z}_+$ (see e.g.\ \cite{majda2}), and $u(x,t)$ and $v(x,t)$ are functions that describe the amplitudes of the mentioned solutions of the LWSEBB equations.
This KdV-like system describes the nonlinear resonant interaction between barotropic and baroclinic Rossby waves.
The properties of the system (\ref{syst}) are studied in \cite{majda3} and \cite{majda1}, where several conservation laws are found and employed for the construction of a perturbed system. Note that an analytic solitary wave solution in the form
$\mathrm{sech}^2(\xi)$ for both components of (\ref{syst}) is found in \cite{majda3}.
In what follows we consider (\ref{syst}) for arbitrary real $d$ unless otherwise explicitly stated.

Quite a few third-order evolution systems, especially those arising in applications, including the celebrated Korteweg--de Vries equation, possess local or nonlocal Hamiltonian structures, cf.\ e.g.\ \cite{ds, Olver, sergy, sergdem, vlad} and references therein. The system (\ref{syst}) is no exception: it is Hamiltonian (we refer the reader to \cite{Olver} for details on the Hamiltonian formalism for PDEs), with an operator $$\mathfrak{D}=\left(\begin{array}{cc}D_x&0\\ 0&D_x\end{array}\right)$$ being the Hamiltonian operator and $\mathcal{H}=-\frac{1}{2}\int (du_x^2+v_x^2+u^2v)\ \mathrm{d}x$ being the corresponding Hamiltonian functional; here $D_x$ denotes the total $x$-derivative
$$D_x=\frac{\partial}{\partial x}+\sum_{j=0}^{\infty}u_{j+1}\frac{\partial}{\partial u_j}+\sum_{j=0}^{\infty}v_{j+1}\frac{\partial}{\partial v_j},$$
where $u_j$ (resp.\ $v_j$) denotes the $j$th derivative of $u$ (resp.\ of $v$) with respect to the spatial variable $x$; $u_0\equiv u$, $v_0\equiv v.$
The above means that the system (\ref{syst}) can be written as
$$(u_t,v_t)^{T}=\mathfrak{D}\delta\mathcal{H},$$
where
$\int\ \mathrm{d}x$ is understood as a formal integral in the sense of calculus of variations (see e.g.\ \cite{dickey, Olver}) and the superscript $T$ indicates the transposed matrix. For any functional $\mathcal{F}=\int f(x,t,u,v,\dots,u_s,v_r)\ \mathrm{d}x$ with a smooth density $f$ 
its variational derivative is defined (see e.g.\ \cite{Olver}) as
$$\delta\mathcal{F}=\left(\frac{\delta f}{\delta u},\frac{\delta f}{\delta v}\right)^T,\ \mbox{where } \frac{\delta }{\delta u}=\sum_{j=0}^{\infty}(-D_x)^j\circ\frac{\partial }{\partial u_j},\  \frac{\delta }{\delta v}=\sum_{j=0}^{\infty}(-D_x)^j\circ\frac{\partial }{\partial v_j}.$$

Existence of the Hamiltonian operator implies (cf.\ e.g.\ \cite{wh_integrability, complete_lists, Olver})
that an infinite set of standard obstacles for existence of
infinitely many conservation laws of increasing order for (\ref{syst}) vanishes,
thus leaving open the problem of description of a complete set of conservation laws for (\ref{syst}).\looseness=-1

Below we describe all generalized symmetries and all nontrivial conservation laws for the Majda--Biello system (\ref{syst}). 

First of all, we prove that this system for $d\neq 0$ is not symmetry integrable: it has no generalized symmetries of order higher than eight.
This allows us to find all generalized symmetries for (\ref{syst}) with $d\neq 0$. Moreover, in fact by our Theorem~\ref{symetrie_veta} for $d\neq 0$ there exist just three symmetries of (\ref{syst}), and all three are equivalent to the Lie point ones.
Note that while using the general results from \cite{meshkov} on the integrability of systems from a certain class that includes (\ref{syst}) one could prove non-existence of {\em time-independent} generalized symmetries for (\ref{syst}) of order greater than five, our result is stronger than that: it shows that the same holds true for any generalized symmetries of order greater than three including {\em explicitly time-dependent} ones.\looseness=-1

Knowing all generalized symmetries in conjunction with the Hamiltonian operator $\mathfrak{D}$ we obtain a complete description of the conservation laws for (\ref{syst}) in Theorem \ref{zakony_veta} below. It turns out that the only conserved functionals are the energy, the Hamiltonian functional and the functionals corresponding to the zonal averages, $\overline{u}(t)=\int u(x,t)\mathrm{d}x$ and $\overline{v}(t)=\int v(x,t)\mathrm{d}x $.  Our result provides {\em inter alia} a rigorous proof of the fact that the system (\ref{syst}) for $d\neq 0$ 
has just the conservation laws found by Majda and Biello in \cite{majda3} and \cite{majda1}.\looseness=-1
\section{Preliminaries}
Below we work in the jet space $J^{\infty}(\pi)$ where $\pi:\mathbb{R}^3\to\mathbb{R}$ is a trivial bundle, with local coordinates $x, u^1, u ^2, u^1_1,u^2_1,u^1_2,u ^2_2,\dots$; here $x$ stands for the spatial variable, $u^{\alpha}$ for the dependent variables and $u^{\alpha}_j$'s correspond to their derivatives $\partial^j u^{\alpha}/\partial x^j$. By $\mathbf{u}^{(n)}$ we mean the set of all derivatives of $u^{1}$ and $u^2$ with respect to $x$ up to order $n$.
Recall \cite{Olver} that a \textit{differential function} is a function $f$ that smoothly depends on $t$, $x$, $u^{\alpha}$'s ($\alpha=1, 2
$)
and a finite number of derivatives of $u^{\alpha}$'s with respect to $x$, i.e. there exists $n\in\mathbb{N}\cup\{0\}$ such that $f=f(x,t,\mathbf{u}^{(n)})$.
The \textit{differential order} of a differential function $f$, denoted by $\mathrm{ord}(f)$, is the maximal $s\in\mathbb{N}\cup\{0\}$ such that $\frac{\partial f}{\partial u^{\alpha}_s}\neq 0$ for some $\alpha\in\{1,2\}$ if $f$ is not a quasiconstant function in the sense of \cite{Kac}, i.e., a function depending only on $x$ and $t$. If $f$ is quasiconstant then $\mathrm{ord}f=-\infty$ by definition. If we are not interested in the differential order of a differential function, we write just $f=f[\mathbf{u}]$.\looseness=-1

Consider a system of $n$-th order evolution equations of the form
\begin{equation}\label{evsyst}
\mathbf{u}_t=\mathbf{F}(x,\mathbf{u}^{(n)}),
\end{equation}
where $\mathbf{u}=(u^1,u^2)^T$, 
and $\mathbf{F}=(F_1, F_2)^T$ is a 2-tuple of explicitly time-independent differential functions.
A 2-tuple of differential functions\footnote{For the sake of simplicity we identify here and below the generalized symmetry with its characteristics, cf.\ e.g.\ \cite{wh_integrability, Olver}.}  $\mathbf{Q}=(Q_1, Q_2)^T$ with $Q_{\alpha}=Q_{\alpha}(x, t,\mathbf{u}^{(k)})$ for all $\alpha =1, 2$, is called a \textit{generalized symmetry of order $k$ for (\ref{evsyst})} if $\mathrm{ord}\ Q_{\alpha}=k$ for at least one $\alpha\in\{1,2\}$,
\begin{equation}\label{defsym}D_t(\mathbf{Q})=\mathrm{D}_{\mathbf{F}}(\mathbf{Q}),\end{equation}
where $D_t=\frac{\partial}{\partial t}+\sum_{\alpha=1}^2\sum_{j}D_x^j(F_{\alpha})\frac{\partial}{\partial u^{\alpha}_j}$, and $\mathrm{D}_{\mathbf{F}}$ is a $2\times 2$ matrix differential operator with entries $(\mathrm{D}_{\mathbf{F}})_{\mu\nu}=\sum_{j}\frac{\partial F_{\mu}}{\partial u^{\nu}_j}D_x^j$, so $(\mathrm{D}_{\mathbf{F}}(\mathbf{Q}))_\mu=\sum_{\nu}(\mathrm{D}_{\mathbf{F}})_{\mu\nu}(Q_\nu)$.

\begin{lm}[\cite{complete_lists}]\label{gs}
If $\mathbf{Q}=(Q_1, Q_2)$ is a generalized symmetry for the system (\ref{evsyst}) then the condition
\begin{equation}\label{podm}\frac{\partial}{\partial t}(\mathrm{D_{\mathbf{Q}}})+\mathrm{pr}\ \mathbf{v}_{\mathbf{F}}(\mathrm{D_{\mathbf{Q}}})+\left[\mathrm{D_{\mathbf{Q}}},\mathrm{D_{\mathbf{F}}}\right]=\mathrm{pr}\ \mathbf{v}_{\mathbf{Q}}(\mathrm{D_{\mathbf{F}}})\end{equation}
holds identically. Here $\frac{\partial}{\partial t}(\mathrm{D_{\mathbf{Q}}})$ is a matrix differential operator with the entries $\left(\frac{\partial}{\partial t}(\mathrm{D_{\mathbf{Q}}})\right)_{\mu\nu}=\frac{\partial }{\partial t}\left(\mathrm{D_{\mathbf{Q}}}\right)_{\mu\nu}$, $\mathrm{pr}\ \mathbf{v}_{\mathbf{F}}(\mathrm{D_{\mathbf{Q}}})$ is a differential operator with the entries $\left(\mathrm{pr}\ \mathbf{v}_{\mathbf{F}}(\mathrm{D_{\mathbf{Q}}})\right)_{\mu\nu}=\mathrm{pr}\ \mathbf{v}_{\mathbf{F}}\left(\mathrm{D_{\mathbf{Q}}}\right)_{\mu\nu}$ where $\mathrm{pr}\ \mathbf{v}_{\mathbf{F}}=\sum_{\alpha,i}D_x^i(F_{\alpha})\frac{\partial}{\partial u^{\alpha}_i}$, the operator $\mathrm{pr}\ \mathbf{v}_{\mathbf{Q}}(\mathrm{D_{\mathbf{F}}})$ is defined in a similar fashion, and $\left[\cdot,\cdot\right]$ is the usual commutator of operators.
\end{lm}
\begin{pz}
Note that from the computational point of view the necessary condition (\ref{podm}) on the symmetry $\mathbf{Q}=(Q_1, Q_2)^T$ for (\ref{evsyst}) is  very often much more helpful than the corresponding relation from the very definition of a symmetry,
see e.g. \cite{complete_lists} for details.

\end{pz}
A \textit{conservation law} for (\ref{evsyst}) is (see e.g.\ \cite{complete_lists, Olver})
the equation $$D_t(\rho(x,t,\mathbf{u}^{(l)}))=D_x(\sigma(x,t,\mathbf{u}^{(s)})),$$ which holds identically.
The differential function $\rho=\rho(x,t,\mathbf{u}^{(l)})$ (resp.\ $\sigma=\sigma(x,t,\mathbf{u}^{(s)})$) is called \textit{the density} (resp.\ the \textit{flux}) of the conservation law in question.

A functional $\mathcal{F}=\int f \mathrm{d}x$ is called {\it conserved} if $f$  is a density of a conservation law for (\ref{evsyst}).

A \textit{cosymmetry} for (\ref{evsyst}) is (see e.g.\ \cite{Blaszak}) a $q$-component differential function $\mathbf{G}$ which satisfies the condition
$$D_t(\mathbf{G})+\mathrm{D}_{\mathbf{F}}^{\ast}(\mathbf{G})=0,$$
where $\left(\mathrm{D}_{\mathbf{F}}^{\ast}\right)_{\mu\nu}=\sum_{i=0}(-1)^iD_x^i\circ\left(\frac{\partial F_{\nu}}{\partial u_i^{\mu}}\right)$ is the formal adjoint of $\mathrm{D}_{\mathbf{F}}$. It is well known (see e.g.\ \cite{Blaszak}) that if $\rho$ is a conservation law density for (\ref{evsyst}), then $\frac{\delta \rho}{\delta \mathbf{u}}$ is a cosymmetry for (\ref{evsyst}).

\section{Main results}
Our goal in this section is to show that the system (\ref{syst}) with $d\neq 0$ has no generalized symmetries (including the time-dependent ones) of order greater than eight (in fact, for the cases $d\neq 1/7$ we find upper bounds for the order of symmetries lower than eight, see Proposition \ref{prp_sym}), and, based on this result, to describe all generalized symmetries and nontrivial conservation laws for (\ref{syst}). Note that since the coefficients at the highest-order derivatives in (\ref{syst}) are constant, it is impossible, in contrast with \cite{vodova}, to apply the method of \cite{sergyeyev} for establishing the nonexistence of higher-order symmetries for (\ref{syst}), so we employ a more direct approach based on Lemma~\ref{gs}.
\begin{prp}\label{prp_sym}
The following assertions hold:
\begin{enumerate}
\item[(i)]{The Majda--Biello system (\ref{syst}) for $d\neq 0,1,1/7$ has no generalized symmetries (including time-dependent ones) of order greater than six.}
\item[(ii)]{The Majda--Biello system (\ref{syst}) for $d=1/7$ has no generalized symmetries (including time-dependent ones) of order greater than eight.}
\item[(iii)]{
The Majda--Biello system (\ref{syst}) for $d=1$ has no generalized symmetries (including time-dependent ones) of order greater than five.}
\end{enumerate}\end{prp}
To simplify notation, 
throughout the proof the vector $\mathbf{u}=(u,v)^T$ from the system (\ref{syst}) will be denoted by $(u^1,u^2)^T$.
\begin{proof}
First we derive a general formula (formula (\ref{polynom})) that follows from Lemma \ref{gs} and which comes in more handy in course of the proof. To this end, suppose that there is a symmetry $\mathbf{Q}=(Q_1,Q_2)^T$ of order $k\geq k_0>0$ for the system (\ref{syst}). Then, according to Lemma~\ref{gs}, the condition
\begin{equation}\label{podminka}\frac{\partial}{\partial t}(\mathrm{D_{\mathbf{Q}}})+\mathrm{pr}\ \mathbf{v}_{\mathbf{F}}(\mathrm{D_{\mathbf{Q}}})+\left[\mathrm{D_{\mathbf{Q}}},\mathrm{D_{\mathbf{F}}}\right]=\mathrm{pr}\ \mathbf{v}_{\mathbf{Q}}(\mathrm{D_{\mathbf{F}}}),\end{equation}
 where $\mathbf{F}=(F_1,F_2)^T=( du^1_{3}-u^2u_1-u ^1u^2_1,u^2_{3}-u^1u^1_1)^T$, holds identically. It can be readily verified that the operator on the right-hand side of (\ref{podminka}) is a first-order matrix differential operator,
$$\mathrm{pr}\ \mathbf{v}_{\mathbf{Q}}(\mathrm{D_{\mathbf{F}}})=\left(\begin{array}{cc}-D_x(Q^2)-Q^2D_x&-D_x(Q^1)-Q^1D_x\\-D_x(Q^1)-Q^1D_x&0\end{array}\right).$$
 Therefore, all the  matrix coefficients at $D_x^m$ on the left-hand side of (\ref{podminka}) with $m\geq 2$ must be equal to zero; we will write this as
\begin{equation}\label{podm2}\frac{\partial}{\partial t}(\mathrm{D_{\mathbf{Q}}})+\mathrm{pr}\ \mathbf{v}_{\mathbf{F}}(\mathrm{D_{\mathbf{Q}}})+\left[\mathrm{D_{\mathbf{Q}}},\mathrm{D_{\mathbf{F}}}\right]=0\ \mbox{mod}\ D_x.\end{equation}

In order to obtain some information about $Q_1$ and $Q_2$,
let us equate the coefficients at  $D_x^{k+3}$, $D_x^{k+2}$, $\dots$, $D_x^{k-k_0+3}$ on the left-hand side of this equation to zero. To simplify computations we shall
use the so-called symbolic transform (see e.g.\ \cite{folland,Gel'fand, sw} and references therein), i.e., we replace a matrix differential operator $\mathfrak{R}=\sum_{j}R_j[\mathbf{u}]D_x^j$ with the formal series $\widetilde{ \mathfrak{R}}=\sum_{j}R_j[\mathbf{u}]z^j$ and employ the formula  \begin{equation}\label{mult}
\widetilde{\mathfrak{R}}\cdot\widetilde{\mathfrak{S}}=\sum_{i=0}^{\infty}\frac{1}{i\mathrm{!}}\frac{\partial^i\widetilde{\mathfrak{R}}}{\partial z^i}D_x^i(\widetilde{\mathfrak{S}})
\end{equation}
for the multiplication of two formal series $\widetilde{\mathfrak{R}}$, $\widetilde{\mathfrak{S}}$ corresponding to differential operators $$\mathfrak{R}=\sum_{i=-\infty}^{l}R_i[\mathbf{u}]D_x^i,\quad \mathfrak{S}=\sum_{i=-\infty}^{m}S_i[\mathbf{u}]D_x^i.$$ 
The formula (\ref{mult}) gives a formal series $\widetilde{\mathfrak{R}}\cdot\widetilde{\mathfrak{S}}$ that corresponds 
to a (pseudo)differential operator $\mathfrak{R}\circ\mathfrak{S}$:
Note that in  (\ref{mult}) we have $D_x^i\left(\sum_{j}R_j[\mathbf{u}]z^j\right)\equiv\sum_{j}D_x^i\left(R_j[\mathbf{u}]\right)z^j$, where $\left(D_x^i(R_j[\mathbf{u}])\right)_{\mu\nu}=D_x^i((R_j[\mathbf{u}])_{\mu\nu})$.

Upon applying the transform in question to (\ref{podm2}) the task of equating the coefficients at  $D_x^{k+3}$, $D_x^{k+2}$, $\dots$, $D_x^{k-k_0+3}$ on the left-hand side of (\ref{podm2}) to zero reduces to that of equating to zero the coefficients at $z^{k+3},z^{k+2},\dots z^{k-k_0+3}$ in
\begin{equation}\label{polynom}\sum_{i=k-k_0}^{k}D_t\left(P_i[\mathbf{u}]\right)z^i+\sum_{i=k-k_0}^{k}P_i[\mathbf{u}]\sum_{j=0}^{k_0}\frac{1}{j\mathrm{!}}\frac{\partial^j z^i}{\partial z^j}D_x^j(\widetilde{\mathrm{D}_{\mathbf{F}}})-
\sum_{i=k-k_0}^{k}P_i[\mathbf{u}]\sum_{j=0}^3\frac{1}{j\mathrm{!}}\frac{\partial^j \widetilde{\mathrm{D}_{\mathbf{F}}}}{\partial z^j}D_x^j(P_i[\mathbf{u}])z^i,\end{equation}
where $P_i[\mathbf{u}]=\left(\begin{array}{cc}\frac{\partial Q_{1}}{\partial u^1_i}&\frac{\partial Q_{1}}{\partial u^2_i}\\ \frac{\partial Q_{2}}{\partial u^1_i}&\frac{\partial Q_{2}}{\partial u^2_i}\end{array}\right)$.
The formula (\ref{polynom}) is just a polynomial in $z$ with matrix coefficients, so that the task of equating its coefficients to zero can be performed e.g.\ using computer algebra software.\looseness=-1

We will now employ (\ref{polynom}) to prove the statements (i)--(iii) from Proposition \ref{prp_sym}.

To prove (i), suppose that there is a symmetry $\mathbf{Q}=(Q_1,Q_2)^T$ of order $k\geq 7$ for the system (\ref{syst}) with $d\neq 0,1/7,1$. Then the coefficients at $z^{k+3},z^{k+2},\dots,z^{k-4}$ in the polynomial (\ref{polynom}) must be equal to zero.
For instance, the coefficients at $z^{k+3}$ and $z^{k+2}$ are equal to $$\hspace*{-1mm}\left(\begin{array}{cc}0& (1-d) (P_{k})_{12}\\
(d-1)(P_{k})_{21} &0
\end{array}\right)\mbox{ and }\left(\begin{array}{cc}-3d D_x\left( (P_{k})_{11}\right)& (1-d) (P_{k-1})_{12}\\
(d-1)(P_{k-1})_{21} &-3D_x\left( (P_{k})_{22}\right)
\end{array}\right),$$ respectively. Since $d\neq 0,1$ by assumption, we obtain $$P_{k}[\mathbf{u}]=\left(\begin{array}{cc}\phi_1(t)&0\\
0&\phi_2(t)\end{array}\right).$$
The coefficients $P_{k-1}[\mathbf{u}],\dots, P_{k-5}[\mathbf{u}]$ can be recursively computed in a similar way by equating the coefficients at $z^{k+3},\dots,z^{k-3}$ to zero. Finally, if we equate the coefficient at $z^{k-4}$ to zero, we obtain the conditions $$D_x\left((P_{k-6})_{11}\right)=f_1,\quad D_x\left((P_{1+r})_{22}\right)=f_2,$$
where $f_1$ and $f_2$ are quite complicated differential functions whose explicit form we omit.

The necessary condition for the above equations to hold is (cf.\ e.g.\ \cite{Olver}) the vanishing of variational derivatives of $f_1$ and $f_2$:
\begin{equation}\label{var}
\delta f_i/\delta u^1=0,\quad \delta f_i/\delta u^2=0,\quad i=1,2.
\end{equation}
However, it turns out that we have, in particular,
\[
\delta f_1/\delta u^1=-k (7d-1)\phi_1(t) u^1 u^1_1/(27(d-1)d^3),\quad \delta f_2/\delta u^2=-2k \phi_2(t) u^1 u^1_1/(9(d-1)).
\]
As $k\neq 0$ by assumption, we see that for $d\neq 0,1,1/7$ the necessary conditions for (\ref{var}) to hold are $\phi_1(t)=0$ and $\phi_2(t)=0$.
But then $\mathbf{Q}=(Q_1,Q_2)^T$ is a symmetry of order $k-1$, which is a contradiction, and the result follows.

To prove (ii), suppose that there is a symmetry  $\mathbf{Q}=(Q_1,Q_2)^T$ of order $k\geq 9$ for the system (\ref{syst}) with $d=1/7$. We substitute $\mathbf{Q}$ and $k_0=9$ into the formula (\ref{polynom}) and equate the coefficients at
$z^{k+3}$, $z^{k+2}$, \dots $z^{k-6}$ to zero. For instance, the coefficients at $z^{k+3}$ and $z^{k+2}$ are equal to
$$\hspace*{-1mm}\left(\begin{array}{cc}0& \frac{6}{7} (P_{k})_{12}\\
-\frac{6}{7}(P_{k})_{21} &0
\end{array}\right)\mbox{ and }\left(\begin{array}{cc}-\frac{3}{7} D_x\left( (P_{k})_{11}\right)& \frac{6}{7} (P_{k-1})_{12}\\
-\frac{6}{7}(P_{k-1})_{21} &-3D_x\left( (P_{k})_{22}\right)
\end{array}\right),$$ respectively. We obtain $$P_{7+r}[\mathbf{u}]=\left(\begin{array}{cc}\phi_1(t)&0\\
0&\phi_2(t)\end{array}\right).$$
The coefficients $P_{k-1},\dots,P_{k-5}$ can be computed in the same fashion by comparing the coefficients at $z^{k+2},z^{k+1},\allowbreak\dots z^{k-3}$ with zero. If we equate the coefficient at $z^{k-4}$ with zero, we obtain $$D_x((P_{k-6})_{22})=g,$$ where $g$ is a differential function whose explicit form we omit. One of the necessary conditions for the above equation to hold is  \begin{eqnarray*}\delta g/\delta u^1&=&\left((7/27)(k-9)+(7/3)\right)u^1u^1_1{\phi}_2=0.
\end{eqnarray*}
Since $k-9\geq 0$, we see that ${\phi}_2(t)=0$.

If we equate the coefficient at $z^{k-6}$ to zero, we obtain $$D_x((P_{k-8})_{11})=f.$$
One of the necessary conditions for this to hold is
\begin{eqnarray*}\delta f/\delta u^1&=&\left((16807/12)+(16807/107)(k-9)\right){\phi}_1u^1u^2_3+\left((6517/4)+(6517/36)(k-9)\right){\phi}_1u^1_1u^2_2\\
&&+\left((6517/4)+(6517/36)(k-9)\right){\phi}_1u^1_2u^2_1+\left((4802/9)+(4802/81)(k-9)\right)u^1u^2u^2_1{\phi}_1=0.\end{eqnarray*}
As $k-9\geq 0$ by assumption, we see that ${\phi}_1(t)=0$ must hold. But this fact together with ${\phi}_2(t)=0$ implies that $\mathbf{Q}=(Q_1,Q_2)^{T}$ is a symmetry of order $k-1$, which is a contradiction, and the result follows.

To prove (iii), suppose that there is a symmetry  $\mathbf{Q}=(Q_1,Q_2)^T$ of order $k\geq 6$ for the system (\ref{syst}) with $d=1$. We substitute $\mathbf{Q}$ and $k_0=6$ into the formula (\ref{polynom}) and equate the coefficients at
$z^{k+3}$, $z^{k+2}$, $\dots, z^{k-3}$ to zero.  In particular, the coefficient at $z^{k+3}$ is a zero matrix, the coefficient at $z^{k+2}$ is equal to
$$\hspace*{-1mm}-3\left(\begin{array}{cc}D_x((P_{k})_{11})&D_x((P_{k})_{12})\\
D_x((P_{k})_{21}) &D_x((P_{k})_{22})
\end{array}\right).$$
This yields
$$P_{k}=\left(\begin{array}{cc}{\phi}_1(t)&{\lambda}_1(t)\\
{\eta}_1(t) &{\psi}_1(t)
\end{array}\right).$$
 The coefficient at $z^{k+1}$  is then equal to $$\left(\begin{array}{cc}-3D_x((P_{k-1})_{11}))-{\lambda}_1u^ 1+{\eta}_1u^1&-{\phi}_1u^1+u^2{\lambda}_1+{\psi}_1u^1-D_x((P_{k-1})_{12})\\
-3D_x((P_{k-1})_{21}) &-3D_x\left( (P_{k-1})_{22}\right)
\end{array}\right).$$ If we equate this coefficient to zero we obtain $(P_{k-1})_{21}={\eta}_2(t)$, $(P_{k-1})_{22}={\psi}_2(t)$, $D_x((P_{k-1})_{11}))=(1/3)(-{\lambda}_1u^1+{\eta}_1u^1)$, $D_x((P_{k-1})_{12}))=(1/3)(-{\phi}_1u^1+{\lambda}_1u^2+{\psi}_1u^1)$. The necessary conditions for the two last formulas to hold are
\begin{eqnarray*}\delta ((1/3)(-{\lambda}_1u^1+{\eta}_1u^1))/ \delta u^1&=&-(1/3)({\lambda}_1-{\eta}_1)=0\\
\delta((1/3)(-{\phi}_1u^1+{\lambda}_1u^2+{\psi}_1u^1))\delta u^1&=&-(1/3)({\phi}_1-{\psi}_1)=0\\
\delta((1/3)(-{\phi}_1u^1+{\lambda}_1u^2+{\psi}_1u^1))\delta u^2&=&(1/3){\lambda}_1=0.
\end{eqnarray*}
This implies ${\lambda}_1(t)={\eta}_1(t)=0$ and ${\psi}_1(t)={\phi}_1(t)$. Hence we have $(P_{k-1})_{12}={\lambda}_2(t)$ and $(P_{k-1})_{1}={\phi}_2(t)$.

The matrices $P_{k-2}$, $P_{k-3}$, and  $P_{k-4}$ can be computed upon equating the coefficients at $z^{k}, \dots z^{k-2}$ to zero. Finally, if we equate the coefficient at $z^{k-3}$ to zero we obtain, in particular,
$$D_x((P_{k-5})_{12})=f,$$
where $f$ is a differential function whose explicit form we omit. One of the necessary conditions for the last formula to hold is
\begin{eqnarray*}
\delta f/\delta u^1&=&\left(-(1/3)-(1/18)(k-6)\right)(u^1)^2{\phi}_1+\left((1/27)(k-6)+(2/9)\right)(u^2)^2{\phi}_1+\alpha(t)=0.
\end{eqnarray*}
As $k-6\geq 0$ by assumption, we have ${\phi}_1(t)=0$, therefore also ${\psi}_1(t)=0$. But then $\mathbf{Q}=(Q_1,Q_2)^T$ is a symmetry of order $k-1$, which is a contradiction, and the result follows.\end{proof}
It follows from Proposition~\ref{prp_sym} that system (\ref{syst}) for $d\neq 0$ is not symmetry integrable in the sense of \cite{symmetries,wh_integrability,complete_lists, test}, i.e., it has no
generalized symmetries of arbitrarily high order. We intend to study the special case of $d=0$ elsewhere.

Now that we know that all generalized symmetries for the system (\ref{syst}) with $d\neq 0$ are of order not higher than eight, they can be found by straightforward computation. The result is given by the following theorem:
\begin{vt}\label{symetrie_veta}
The only generalized symmetries for the Majda--Biello system (\ref{syst}) with $d\neq 0$ are those with characteristics 
$\mathbf{Q}^1=(\frac{x}{3}u_x+\frac{2}{3}u+t(du_{xxx}-vu_x-uv_x),\frac{x}{3}v_x+\frac{2}{3}v+t(v_{xxx}-uu_x))^T$, $\mathbf{Q}^2=(du_{xxx}-vu_x-uv_x,v_{xxx}-uu_x)^T$, and
$\mathbf{Q}^3=(u_x,v_x)^T$, i.e., the scaling symmetry, $t$-translations and $x$-translations.
\end{vt}
\begin{pz}
Using (\ref{syst}) the above symmetries can be written as $\mathbf{Q}^1=(\frac{x}{3}u_x+\frac{2}{3}u+tu_t,\frac{x}{3}v_x+\frac{2}{3}v+tv_t)^T$, $\mathbf{Q}^2=(u_t,v_t)^T$, and $\mathbf{Q}^3=(u_x,v_x)^T$, which facilitates their interpretation. 
In particular, we immediately see that all generalized symmetries for (\ref{syst}) with $d\neq 0$ are equivalent to the Lie point ones.
\end{pz}

It is readily checked that there exists no conserved functional $\mathcal{K}$ associated (through the Hamiltonian operator $\mathfrak{D}$ so that $\mathbf{Q}_1=\mathfrak{D}\delta\mathcal{K}$) to the first symmetry $\mathbf{Q}_1$. The conserved functional associated to the second symmetry $\mathbf{Q}_2$ is $\mathcal{H}=-1/2\int (du_x^2+v_x^2+u^2 v)\ \mathrm{d}x$, the conservation law associated to the third symmetry $\mathbf{Q}_3$ is the energy $1/2\int (u^2+v^2)\ \mathrm{d}x.$

In particular, it follows from Theorem \ref{symetrie_veta} that the highest possible order of a generalized symmetry for (\ref{syst}) with $d\neq 0$ is at most three. We will now use this fact to find all conservation laws for the system (\ref{syst}). First of all we state the following lemma:
\begin{lm}\label{lemma}
The only cosymmetries for the Majda--Biello system (\ref{syst}) with $d\neq 0$ are $\mathbf{G}^1=(1,0)^T$, $\mathbf{G}^2=(0,1)^T$,  $\mathbf{G}^3=(u,v)^T$ and  $\mathbf{G}^4=(du_{xx}-uv,v_{xx}-\frac{1}{2}u^2)^T$.
\end{lm}
\begin{proof}
In order to find all cosymmetries for the system (\ref{syst}), we have to find the maximal order of all cosymmetries. To this end we make use of the fact that our system is Hamiltonian with the Hamiltonian operator $\mathfrak{D}=\left(\begin{array}{cc}D_x&0\\0&D_x\end{array}\right)$ and $\mathfrak{D}$ maps cosymmetries to symmetries. So, let $\mathbf{G}=(G_1,G_2)^T$ be a cosymmetry for (\ref{syst}). Then the 2-tuple $\mathfrak{D}(\mathbf{G})=(D_x(G_1),D_x(G_2))^T$ is a symmetry. According to the previous results, the highest possible order of  $\mathfrak{D}(\mathbf{G})$ is equal to 3. Therefore, $\mathbf{G}=(G_1,G_2)^T$ must be of differential order at most two. Now, it is just a matter of straightforward computation to find all cosymmetries for (\ref{syst}).\looseness=-1
\end{proof}
\begin{vt}\label{zakony_veta}
The only linearly independent conservation laws for the Majda--Biello system $(\ref{syst})$  with $d\neq 0$ are, modulo the addition of trivial conservation laws, of the form
$D_t(\rho)=D_x(\sigma)$, where 
$\rho$ (resp.\ $\sigma$) are linear combinations of
$\rho_i$ (resp.\ $\sigma_i$), $i=1,\dots,4$, given by the formulas
$$\begin{array}{llllll}
\rho_1&=&u,&\sigma_1&=&du_{xx}-uv,\\[5mm]
\rho_2&=&v,&\sigma_2&=&v_{xx}-\frac{1}{2}u^2\\[5mm]
\rho_3&=&u^2+v^2,&\sigma_3&=&2duu_{xx}-2u^2v-du_x^2+2vv_{xx}-v_x^2\\[5mm]
\rho_4&=&du_x^2+v_x^2+vu^2,\ \ &\sigma_4&=&2v_x(v_{xxx}-uu_x)+2du_x( du_{xxx}-vu_x-uv_x)-v_{xx}^2-d^2u_{xx}^2\\[5mm]
&&&&&+u^2v_{xx}+2duvu_{xx}-\frac{1}{4}u^2(u^2+v^2).
\end{array}$$
\end{vt}
\begin{proof}
Let $\rho$ be a conservation law density for (\ref{syst}). Then the 2-tuple $\frac{\delta\rho}{\delta \mathbf{u}}=(\frac{\delta\rho}{\delta u},\frac{\delta\rho}{\delta v})$ is a cosymmetry for (\ref{syst}). From Lemma \ref{lemma} it follows that $\mathrm{ord}\ \frac{\delta\rho}{\delta u}\leq 2$ and $\mathrm{ord} \frac{\delta\rho}{\delta v}\leq 2$, which means that up to the addition of a trivial density $\rho$ is a function of $x,t,u,v,u_x,v_x$ only. In order to find all conserved densities for (\ref{syst}) we now need to find all solutions to the system
\begin{eqnarray}\label{rho}
 \frac{\delta\rho(x,t,u,v,u_x,v_x)}{\delta u}&=&a_4(du_{xx}-uv)+a_3 u+a_1,\\
 \frac{\delta\rho(x,t,u,v,u_x,v_x)}{\delta v}&=&a_4\left(v_{xx}-\frac{1}{2}u^2\right)+a_3 v+a_2,
\end{eqnarray}
where the right-hand sides are linear combinations of the components of cosymmetries for (\ref{syst}), so $a_i$ are arbitrary constants.

The general solution of (\ref{rho}) modulo the trivial conserved density
is
\begin{eqnarray*}\rho(x,t,u,v,u_x,v_x)&=& \sum\limits_{i=1}^4 a_i \rho_i,
\end{eqnarray*}
and the result follows. \looseness=-1
%
The fluxes $\sigma_i$ are easy to find
by straightforward computation.
\end{proof}
The conserved functionals $\int \rho_i \mathrm{d}x$, $i=1,2$, corresponding to the first and second conserved densities $\rho_1$ and $\rho_2$ are Casimir functionals corresponding to the Hamiltonian operator $\mathfrak{D}$. They were used in \cite{majda1} for the construction of a perturbed system for (\ref{syst}). The conserved functional $\int \rho_3 \mathrm{d}x$ corresponding to the third conserved density $\rho_3$ is the energy, i.e., the integral of motion associated with the invariance under the time shifts.
The conserved functional corresponding to the fourth conserved density is just the Hamiltonian functional $\mathcal{H}$ which generates the dynamics of (\ref{syst}). Note that both equations of (\ref{syst}) have the form of conservation laws, so the system (\ref{syst}) is written in the normal form in the sense of \cite{ps}.

The above conservation laws of (\ref{syst}) can be employed e.g.\ for the stability analysis of (\ref{syst}) and
for the construction of nonlocal symmetries and nonlocal conservation laws for (\ref{syst}) through introduction of the associated potentials, cf.\ e.g.\ \cite{bocharov, kunzinger} and references therein. Also, the knowledge of all symmetries and conservation laws for (\ref{syst}) enables one to construct higher-precision discretizations of (\ref{syst}), cf.\ e.g.\ \cite{lw}.

\section*{Acknowledgements}
The author thanks Dr. A. Sergyeyev for stimulating discussions and Prof. I.
 S. Krasil'shchik for helping her to gain a deeper understanding of the
 subject. This research was in part supported by the fellowship from the
 Moravian-Silesian region and the institutional support for I\v{C}47813059.
\looseness=-1


\begin{thebibliography}{99}
\bibitem{majda3} J. Biello, A. J. Majda, \textit{The effect of meriditional and vertical shear on the interaction of equatorial baroclinic and barotropic Rossby waves}, Stud. Appl. Math. 112 (2004), 341--390.
\bibitem{Blaszak} M. Blaszak, \textit{Multi-Hamiltonian Theory of Dynamical Systems}, Springer, Berlin etc., 1998.
\bibitem{bocharov} A. Bocharov et al, \textit{Symmetries and Conservation Laws for Differential Equations of Mathematical Physics}, AMS, Providence, R1, 1999.
\bibitem{burde}G. I. Burde, A. Sergyeyev, \textit{Ordering of two small parameters in the shallow water wave problem}, J. Phys. A: Math. Theor. 46 (2013), paper 075501, arXiv:1301.6672.
\bibitem{Kac} A. De Sole, V. Kac, M. Wakimoto, \textit{On classification of Poisson vertex algebras}, Transformation Groups  15 (2010), 883--907, arXiv:1004.5387.
\bibitem{dickey} L. A. Dickey, \textit{Lectures on classical W-algebras}, Acta Appl. Math. 47 (1997), 243--321.
\bibitem{ds}V. G. Drinfeld, V.V. Sokolov, {\it Lie algebras and equations of Korteweg--de Vries type}, J. Sov. Math. 30 (1985), 1975--2036.
\bibitem{folland} G. B. Folland, \textit{Introduction to partial differential equations}, Princeton University Press, 1995.
\bibitem{Gel'fand} I. M. Gel'fand, L. A. Dikii, \textit{Fractional powers of operators and Hamiltonian systems}, Func. Anal. Appl. 10 (1976), 259--273.
\bibitem{hoskins} B. J. Hoskins, F.-F. Jin, \textit{The initial value problem for tropical perturbations to a baroclinic atmosphere}, Quart J. Roy. Meteor. Soc. 117 (1991), 299--317.
\bibitem{kiladis} G. Kiladis, M. Wheeler, \textit{Horizontal and vertical structure of observed tropospheric equatorial Rossby waves}, J. Geophys. Res. 100 (1995), 22981--22997.
\bibitem{kunzinger}  M. Kunzinger, R.O. Popovych, \textit{Potential conservation laws}, J. Math. Phys. 49 (2008), no.~10, paper 103506, arXiv:0803.1156.
\bibitem{lw}D. Levi, P. Winternitz, {\it Continuous symmetries of difference equations}, J. Phys. A: Math. Gen. 39 (2006), R1--R63, arXiv:nlin/0502004.
\bibitem{majda2} B. Khouider, A. J. Majda, S. N. Stechmann, \textit{Climate science in the tropics: waves, vortices and PDEs}, Nonlinearity 26 (2013), R1--R68.
\bibitem{majda4} A. J. Majda, \textit{Introduction to PDEs and Waves for the Atmosphere and Ocean}, 
NYU, CIMS, New York; AMS, Providence, RI, 2003. 
\bibitem{majda1} A. J. Majda, J. A. Biello, \textit{The Nonlinear Interaction of Barotropic and Equatorial Baroclinic Rossby Waves}, Journal of the atmospheric sciences, 60, no. 15 (2003), 1809-1821.
\bibitem{majda5} A. J. Majda, M. Shefter, \textit{Models for stratiform instability and convectively coupled waves}, J. Atmos. Sci. 58 (2001), 1567--1584.
\bibitem{majda6} A. J. Majda, R. Klein, \textit{Systematic multiscale models for the Tropics}, J. Atmos. Sci. 60 (2003), 357--372.
\bibitem{meshkov}A. G. Meshkov, \textit{On symmetry classification of third-order evolutionary systems of divergent type}, J. Math. Sci. 151 (2008), no.~4, 3167-3181.  
\bibitem{symmetries} A. V. Mikhailov, V. V. Sokolov, \textit{Symmetries of Differential Equations and the Problem of Integrability}, in \textit{Integrability}, A. V. Mikhailov (ed.), 19--98, Springer, Berlin etc., 2009.
\bibitem{wh_integrability} A. V. Mikhailov, A. B. Shabat, V. V. Sokolov, \textit{Thy Symmetry Approach to Classification of integrable equations}, in \textit{What is Integrability?}, V. E. Zakharov (ed.), 115--184,
Springer, Heidelberg etc., 1991.
\bibitem{complete_lists} A. V. Mikhailov, A. B. Shabat, R. I. Yamilov, \textit{The symmetry approach to the classification of non-linear equations. Complete lists of integrable systems}, Russian Math. Surveys 42 (1987),no.~4, 1--63.
\bibitem{neelin} J. D. Neelin, N. Zeng, \textit{A quasi-equilibrium tropical circulation model - Formulation}, J. Atmos. Sci. 57 (2000), 1741--1766.
\bibitem{Olver} P. J. Olver, \textit{Applications of Lie Groups to Differential Equations}, Springer, N. Y., 1993.
\bibitem{ps}R.O. Popovych and A. Sergyeyev, \textit{Conservation laws and normal forms of evolution equations}, Phys. Lett. A 374 (2010), no.~22, 2210--2217, arXiv:1003.1648.
\bibitem{sw}J. A. Sanders, J. P. Wang,
\textit{Number theory and the symmetry classification of integrable systems}, in \textit{Integrability}, A.V. Mikhailov (ed.), 89--118,
Lecture Notes in Phys., 767, Springer, Berlin etc., 2009.
\bibitem{sergy}A. Sergyeyev, \textit{A strange recursion operator demystified}, J. Phys. A: Math. Gen. 38 (2005),  L257--L262, arXiv:nlin.SI/0406032.
\bibitem{sergyeyev} A. Sergyeyev, \textit{On time-dependent symmetries and formal symmetries of evolution equations}, in \textit{Symmetry and perturbation theory (Rome, 1998)}, G. Gaeta (ed.), 303--308, World Scientific 1999, arXiv:solv-int/9902002.
\bibitem{sergdem} A. Sergyeyev, D. Demskoi, \textit{Sasa--Satsuma (complex modified Korteweg--de Vries II) and the complex sine-Gordon II equation revisited: Recursion operators, nonlocal symmetries, and more}, J. Math. Phys. 48 (2007), no.~4, paper 042702, arXiv:nlin/0512042.
\bibitem{test} A. B. Shabat, A. V. Mikhailov, \textit{Symmetries - Test of Integrability}, in {\em Important developments in soliton theory}, 355--374,  
    Springer, Berlin etc., 1993.
\bibitem{sokolov} V. V. Sokolov, \textit{On the symmetries of evolution equations}, Russian Math. Surveys 43 (1988), no. 5, 165--204.
\bibitem{vodova}J. Vodov\'a, \textit{A complete list of conservation laws for non-integrable compacton equations of $K(m, m)$ type}, Nonlinearity 26 (2013), 757--762, arXiv:1206.4401.
\bibitem{wang} B. Wang, X. Xie, \textit{Low-frequency equatorial waves in vertically sheared zonal flow. Part I: Stable waves}, J. Atmos. Sci., 53 (1996), 449--467.
\bibitem{vlad}V. A. Vladimirov, C.M\c{a}czka, A. Sergyeyev, S. Skurativskyi, \textit{Stability and dynamical features of solitary wave solutions for a hydrodynamic-type system taking into account nonlocal effects}, Comm. Nonlin. Sci. Numer. Simul. 19 (2014), no.~6, 1770--1782, arXiv:1207.6198.
\end{thebibliography}
\end{document}